\documentclass[11pt]{article}
\usepackage[letterpaper]{geometry}
\usepackage[parfill]{parskip}
\usepackage{amsmath,amsthm,amssymb,bbm}
\usepackage{mathtools}
\usepackage{cases}
\usepackage{graphicx}
\usepackage{tabularx}
\usepackage{microtype}
\usepackage{enumitem}

\usepackage{authblk}

\usepackage{url}
\usepackage[colorlinks,citecolor=blue,urlcolor=blue,linkcolor=blue,linktocpage=true]{hyperref}
\usepackage{cleveref}
\crefformat{equation}{(#2#1#3)}
\crefrangeformat{equation}{(#3#1#4) to~(#5#2#6)}
\crefname{equation}{}{}
\Crefname{equation}{}{}

\usepackage[authoryear]{natbib}

\newtheoremstyle{mythmstyle}
  {8 pt} 
  {3 pt} 
  {} 
  {} 
  {\bfseries} 
  {.} 
  {.5em} 
  {} 

\theoremstyle{mythmstyle}

\newtheorem{theorem}{Theorem}
\newtheorem{lemma}[theorem]{Lemma}
\newtheorem{corollary}[theorem]{Corollary}
\newtheorem{remark}{Remark}
\newtheorem{definition}{Definition}
\crefname{definition}{\textbf{definition}}{definitions}
\Crefname{definition}{Definition}{Definitions}
\crefname{assumption}{\textbf{assumption}}{assumptions}
\Crefname{assumption}{Assumption}{Assumptions}

\begin{document}

\title{Network Cross-Validation for Determining the Number of Communities in Network
Data}

\author[1]{Kehui Chen}
 \author[2]{Jing Lei}
\affil[1]{Department of Statistics, University of Pittsburgh}
\affil[2]{Department of Statistics, Carnegie Mellon University}

\maketitle

\begin{abstract}
The stochastic block model and its variants have been a popular tool
in analyzing large network data with community structures.
In this paper we develop an
efficient \emph{network cross-validation} (NCV) approach to determine the number of communities,
as well as to choose between the regular stochastic block model and the degree corrected
block model. The proposed NCV method is based on a \emph{block-wise node-pair splitting} technique,
combined with an integrated step of community recovery using sub-blocks of the adjacency matrix.
We prove that the probability of under selection vanishes as the number of node increases, under
mild conditions satisfied by a wide range of popular community recovery algorithms. The solid performance of our method is also demonstrated in extensive simulations and a data example.

\end{abstract}


\section{Introduction}
\label{sec:introduction}
In the last few decades, the amount of
network data and the need for relevant statistical inference tools
are growing at a rapid pace.
One of the main research topics in network data analysis is to identify
hidden communities from a single observed network.
Roughly speaking, network community refers to the phenomenon that
individuals close to each other are more likely to connect, and hence
the edge density varies from within coherent subpopulations to between
subpopulations \citep{NewmanG04,Newman06}.
 The stochastic block model
\citep{Holland83} and
its variants such as the degree
corrected block model \citep{KarrerN11} are powerful and
mathematically elegant tools to model large networks with community
structures, and have been proved useful
in many scientific areas such as social science, biology, and information
science \citep{FaustW92,Kemp06,BickelC09}.

The community recovery problem for stochastic block models has been the focus of much
research effort in the past decade, in
several areas including statistics
\citep{BickelC09,ZhaoLZ12,Jin12,Fishkind13,LeiR14},
machine learning \citep{McSherry01,ChenSX12,ChaudhuriCT12,Anandkumar14},
statistical physics \citep{Decelle11,Krzakala13}, and probability theory
\citep{Massoulie13,MosselNS13,AbbeBH14}.  These methods are based on a wide range
of different tools such as maximum likelihood, convex optimization, spectral
methods, and belief propagation, etc.  However, almost all of these methods
require $K$, the total number of communities, to be known in advance.

Unlike the community recovery problem, determining the number of communities remains a challenging problem and gains increasing interest recently. \cite{ZhaoLZ11PNAS} propose to
sequentially extract one significant community from the remaining of the network, and
they approximate the null distribution of their optimizing statistic by
bootstrapping from an Erd\H{o}s-R\'{e}nyi graph. \cite{BickelS13test} propose to test
$K=1$ vs $K>1$ at each step of a recursive bipartition algorithm. They derive the
asymptotic null distribution of the largest eigenvalue of the suitably scaled and
centered adjacency matrix. But the convergence rate is slow and an empirical tuning is
needed in practice. Moreover, these sequential or recursive testing procedures only
work for certain types of community structures. After the first draft of this work, there have been some new developments along the line of testing $K=\tilde K$ for a given candidate value  $\tilde K$ \citep{lei2014goodness}.
Some model selection criteria based on approximated likelihood or Bayesian inference have also been proposed, including \cite{latouche2012variational,mcdaid2013improved,Saldana2014how,wang2015likelihood}.
The performance of these methods often crucially depends on the choice of a penalty term or some prior knowledge of the model parameters.

In this paper, we focus on a generic idea of network cross-validation. Cross-validation is a very popular and appealing method in many model selection problems. The adaptation to network data is usually through a node splitting procedure and has been considered by \cite{airoldi08, neville2012correcting}, among others. A random node-pair splitting method has been used in \cite{hoff2008modeling} for model selection under a Bayesian framework. These methods, even though applicable to the community recovery problem, are usually computationally intensive. More detailed discussion and comparison can be found in \Cref{subsec:splitting}. Moreover, the theoretical investigation for cross-validation methods in network model selection remains open.

The NCV method developed in this paper is based on a \emph{block-wise node-pair splitting} technique.  The splitting step divides the nodes randomly into two groups $\mathcal{N}_1$ and $\mathcal{N}_2$. The observed edge formation between node pairs $(i,j)$, for $i\in \mathcal{N}_1$, and $j\in \mathcal{N}_1 \cup \mathcal{N}_2$ are used as a fitting set, and the node pairs $(i,j)$ for $i,j \in \mathcal{N}_2$ are used as a testing set.
Such a node-pair splitting is superior to a simple 
node splitting. It originates from two key
observations. First, the fitting set carries substantial relationship information for all nodes in the network. That is, we can consistently estimate the membership of all the nodes as well as the
community-wise edge probability matrix from the fitting set.  Second, given
the community membership, the edge formation in the fitting set and in the testing set are independent. 
The second observation reflects a significant difference between the traditional parameterization
of the stochastic block model that treats the node memberships as missing variables,
and the conditional parameterization that treats the memberships as parameters.
The traditional parameterization can model networks of arbitrary size and
has a motivation from exchangeable random graphs \citep{BickelC09}.
However, for community recovery based on a single observed realization of the stochastic block model,
the useful information for statistical inference is largely contained in
the randomness of edge formation.

We describe the algorithm in detail in \Cref{sec:sbm-ncv}. The proposed V-fold NCV method is novel and is of substantial practical interest for several reasons. First, it is computationally efficient, requiring only one model fitting for each fold. Second, it is tuning free except the number of folds. Moreover, it is general enough to be combined with different community recovery techniques.
In \Cref{subsec:theorem}, we characterize the theoretical properties of the proposed NCV method. We show that under appropriate conditions, when combined with popular community recovery techniques, such as modularity based optimization and spectral clustering, the proposed NCV does not underestimate the number of communities with probability tending to one.  The protection
against overestimating is also discussed.
In \Cref{sec:experiment}, we demonstrate the effectiveness of our method via extensive simulations, where different types of network community structures
are investigated.

The NCV method can be applied to select the best model from a general collection of candidate
models, which does not need to be nested or hierarchical. For example,
one can use NCV to choose between
the regular stochastic block model and the degree corrected block model,
with simultaneous choice of number of communities. Moreover, the block-wise node-pair
splitting idea behind NCV can be further extended to other network models with conditional
edge independence. These extensions are described in \Cref{sec:dcbm} and \Cref{sec:discussion}, and illustrated in an
application to a
political blog data in \Cref{sec:experiment}, where the NCV method chooses
 the degree corrected block model with two communities, matching pervious findings in the literature. 
 Further discussions can be found in \Cref{sec:discussion}.

\section{Network cross-validation for stochastic block models}
\label{sec:sbm-ncv}
In a stochastic block model with $n$ nodes and $K$ communities,
the observed random graph is often represented by a
$n$ by $n$ symmetric binary adjacency matrix $A$.
The community structure is represented by a vector $g\in \{1,...,K\}^n$ with
$g_i$ being the community that node $i$ belongs to.  Given the membership vector $g$,
each edge $A_{ij}$ ($i<j$) is an
independent Bernoulli variable satisfying
\begin{equation}
  \label{eq:sbm}
P(A_{ij}=1) = 1-P(A_{ij}=0) = B_{g_i g_j}\,,
\end{equation}
where $B\in [0,1]^{K\times K}$ is a symmetric matrix representing the community-wise
edge probability.  In this section we focus on the problem of estimating $K$,
the number of communities, from a single observed network $A$.
Generalization to
other model selection problems is straightforward and will be discussed in
later sections.

\subsection{Block-wise node-pair splitting}\label{subsec:splitting}

Let $(\mathcal N_1,\mathcal N_2)$ be a random partition of the nodes,
the adjacency matrix can be written in a block form \begin{equation}\label{eq:A-block}
  A=\left(\begin{array}{cc}
A^{(11)} & A^{(12)}\\
A^{(21)} & A^{(22)}
\end{array}\right)\,,\end{equation}
where $A^{(jj)}$ is the adjacency matrix for nodes in $\mathcal N_j$ ($j=1,2$).
The splitting step puts node pairs in $A^{(11)}$ and $A^{(12)}$ in the fitting sample and puts node pairs in
$A^{(22)}$ in the testing sample. Such a split makes full use of the entire observed adjacency matrix and provides a way to directly compare multiple candidate values of $K$ based on the predictive loss on the testing sample.

The block-wise node-pair splitting is novel, and has several appealing features when compared with existing cross-validation methods for network data that are mostly based on a node splitting technique.
In the node splitting method, where
the nodes, instead of the node pairs, are split into a fitting set and a testing set, one typically uses the parameterization of the membership distribution $\pi=(\pi_1,...,\pi_K)$
such that the node membership is generated independently with probability
$P(g_i = k)=\pi_k$ for $1\le i\le n$. After the node splitting, the model parameter
 $(\pi,B)$ is estimated on the subnetwork
confined on the fitting set of nodes, and 
evaluated on the subnetwork confined on the testing
subset of nodes.
This approach has some drawbacks. First, calculating the full likelihood in terms of $\pi$ and $B$ in presence of a missing
membership vector $g$ is computationally demanding. Second, it does not use the observed edge formation between
the fitting and testing nodes, introducing unnecessary randomness in the validation
step by treating the node memberships as random variables.
In contrast, the block-wise node-pair splitting allows us to use the conditional parameterization that treats the memberships as parameters and estimate,  from the fitting set of node pairs, the membership $g$ for all nodes in the entire network.  Such a procedure fully exploits the information carried in node pairs between $\mathcal N_1$ and $\mathcal N_2$, making the model fitting and validation statistically and computationally more efficient.

Comparing with the uniformly random node-pair splitting method, the proposed block-wise splitting technique is much simpler to implement as many existing community recovery methods can be easily extended to the rectangular adjacency matrix $(A^{(11)},A^{(12)})$ (\Cref{subsec:estimation}). In contrast, the random node-pair splitting will result in incomplete adjacency matrices, and users are left with limited choices which are mainly based on likelihood.

%

\subsection{Estimating model parameters from the rectangular matrix}\label{subsec:estimation}
After splitting, we estimate model parameters $(g, B)$ from the $n_1\times n$
 rectangular matrix
$A^{(1)}=(A^{(11)}, A^{(12)})$, where $n_1$ is the cardinality of $\mathcal N_1$.  Many standard procedures designed for
the full adjacency matrix can be extended to this case, such as likelihood
based methods and spectral methods.  Here we focus on spectral clustering, because
it is simple to implement and the analysis is straightforward.  This is also
the method we implement in the numerical experiments presented in \Cref{sec:experiment}.

For a given candidate value $\tilde K$ of $K$, the simple spectral clustering method first performs a singular value
decomposition
on $A^{(1)}$, and estimates $g$  by applying $k$-means clustering on the
rows of
the $n\times \tilde K$ matrix consisting of the leading $\tilde K$ right singular vectors.
Once $\hat g$ is obtained, let $\mathcal N_{j,k}$
be the nodes
in $\mathcal N_j$ with estimated membership $k$, and
 $n_{j,k}=|\mathcal N_{j,k}|$ ($j=1,2$, $1\le k\le \tilde K$).
 We can estimate $B$ using a simple plug-in estimator:
\begin{equation}\label{eq:B-est}
\hat B_{k,k'} =\left\{\begin{array}{ll}
 {\displaystyle\frac{\sum_{i\in \mathcal N_{1,k}, j\in\mathcal N_{1,k'}\cup \mathcal N_{2,k'}}
 A_{ij}}{n_{1,k}(n_{1,k'}+n_{2,k'})}}\,, & k\neq k'\,,\\
 &
 \\
 {\displaystyle
 \frac{\sum_{i,j\in \mathcal N_{1,k}, i<j} A_{ij} +
 \sum_{i\in\mathcal N_{1,k}, j\in\mathcal N_{2,k}}A_{ij}}{
 (n_{1,k}-1)n_{1,k}/2+n_{1,k}n_{2,k}}
 }
 \,,& k=k'\,.
 \end{array}\right.
\end{equation}

\subsection{Validation using the testing set}\label{subsec:validate}
After estimating the parameters $(\hat g,\hat B)$, we can
assess the goodness-of-fit by validating
on the testing set.

For each observation in the testing set, $A_{ij}$ ($i\neq j$, $i,j\in\mathcal N_2$)
is a Bernoulli random variable with parameter $P_{ij}=B_{g_i g_j}$, which is estimated
by $\hat P_{ij}=\hat B_{\hat g_i \hat g_j}$.  Some natural choices of
the loss function $\ell$ include negative log-likelihood
$\ell(x,p)=-x\log p-(1-x)\log(1-p)$, and squared error $\ell(x,p)=(x-p)^2$.
In our numerical experiments, these two loss functions give almost identical
performance.

In the validation step, if the candidate value $\tilde K$ is too small, then the fitted model cannot capture the
fine structures in the data, and will likely lead to poor predictive loss on testing data.
If $\tilde K$ is too large, then the model over fits the data, with noisy prediction
on the testing data. Therefore, it is natural to expect the validated
predictive loss $\hat L(A,\tilde K) = \sum_{i,j\in\mathcal N_2,i\neq j}
 \ell(A_{ij}, \hat P_{ij})\,,$ to be minimized when $\tilde K=K$, the true number
of communities. Partial theoretical supports and further heuristic arguments are given in \Cref{subsec:theorem}.

\subsection{V-fold network cross-validation}\label{subsec:ncv}
Now we formally describe the V-fold network cross-validation procedure.

{\bf Algorithm 1: V-fold network cross-validation}
\begin{enumerate}[itemsep=-3pt,topsep=0pt]
  \item [] \textbf{Input:} adjacency matrix $A$, a set $\mathcal K$ of
  candidate values for $K$, number of folds
  $V\ge 2$.
  \item Randomly split the adjacency matrix into  $V\times V$ equal
  sized blocks
  $$
  A = (A^{(uv)}:1\le u,v\le V)
  $$ similarly as in \eqref{eq:A-block},
  where the nodes are partitioned into $V$ equal-sized subsets
  $\mathcal N_v$ ($1\le v\le V$);
   $A^{(vv)}$ contains observation between node pairs in the $v$th random subset
   $\mathcal N_v$; and $A^{(uv)}$ contains observation between $\mathcal N_v$
  and $\mathcal N_v$.
  \item For each $1\le v\le V$, and each $\tilde K\in \mathcal K$
   \begin{enumerate}[itemsep=-3pt,topsep=0pt]
     \item Estimate model parameters $(\hat g^{(v)}, \hat B^{(v)})$ using
     the rectangular
  submatrix obtained by removing the rows of $A$ in subset $\mathcal N_v$
  $$
  A^{(-v)}=(A^{(rs)}:r\neq v,\ 1\le r, s\le V)\,.$$
  \item Calculate the predictive loss evaluated on $A^{(vv)}$:
  $$
  \hat L^{(v)}(A,\tilde K) = \sum_{i,j\in\mathcal N_v,i\neq j}
  \ell(A_{ij}, \hat P^{(v)}_{ij})\,,
  $$
  where $\hat P^{(v)}_{ij}=\hat B^{(v)}_{\hat g^{(v)}_i,\hat g^{(v)}_j}$.
  \end{enumerate}
  \item Let $\hat L(A,\tilde K)=\sum_{v=1}^V \hat L^{(v)}(A,\tilde K)$ and
  output
  $$
  \hat K = \arg\min_{\tilde K\in\mathcal K} \hat L(A,\tilde K)\,.
  $$
\end{enumerate}
In our experiments we found the performance of NCV insensitive to the choice
of $V$, and we used $V=3$ for all numerical experiments.  Further
discussion on the choice of $V$ and its difference from the regular
cross-validation is given in \Cref{sec:discussion}.

\subsection{Theoretical properties}\label{subsec:theorem}
For two sequences $a_n$ and $b_n$,
we denote $a_n=\Omega(b_n)$ if $b_n=O(a_n)$, and $a_n=\omega(b_n)$ if $b_n=o(a_n)$.

To study the asymptotic behavior of the estimator, we consider
a sequence of SBM's parameterized by $(g^{(n)},B^{(n)}:n\ge 1)$ such that
\begin{enumerate}
  \item [(A1)]
  $g^{(n)}$ is a membership vector of length $n$ with $K$ distinct communities, and the minimum block size is at least $\pi_0 n$ for some constant $\pi_0$.
  \item [(A2)] $B^{(n)}=\rho_n B_0$ where $B_0$ is a $K\times K$ symmetric matrix with entries
  in $(0,1]$, and the rows of $B_0$ are all distinct. The rate $\rho_n$ controls the largest
  edge probability and satisfies $\rho_n=\Omega(\log n /n)$.
\end{enumerate}
The first condition requires the block sizes to be relatively balanced. It is satisfied
with high probability if the memberships are independently generated from a multinomial distribution.
The second condition allows the edge probability to decrease at a rate $\rho_n$ as
$n$ increases. The lower bound on the sparsity $\rho_n$ makes it possible to obtain accurate community
recovery.

The community recovery is an integrated part in the proposed NCV method
and its accuracy plays an important role in the performance of model selection.
We  introduce two notions of community recovery consistency.
\begin{definition}[Exactly consistent recovery]
Given a sequence of SBM's with $K$ blocks parameterized by $(g^{(n)},B^{(n)})$,
we call
a community recovery method $\hat g$ \emph{exactly consistent}
if $P(\hat g(A,K) = g^{(n)})\rightarrow 1$, where $A$ is a realization of SBM $(g^{(n)},B^{(n)})$ and the equality is up to a possible label permutation.
\end{definition}

\begin{definition}[Approximately consistent recovery]
  For a sequence of SBM's with $K$ blocks parameterized by $(g^{(n)},B^{(n)})$ and
  a sequence $\eta_n=o(1)$, we say $\hat g$ is approximately consistent with rate $\eta_n$ if,
    $$\lim_{n\rightarrow\infty}P[{\rm Ham}(\hat g(A,K), g)\ge \eta_n n]= 0,$$
  where ${\rm Ham}(\hat g,g)$ is the smallest Hamming distance between $\hat g$ and $g$
  among all possible label permutations.
\end{definition}

Exactly consistent community recovery can be achieved under mild assumptions on
$g^{(n)}$ and $B^{(n)}$. It is known that likelihood methods \citep{BickelC09} are exactly consistent
when $\rho_n n/\log n \rightarrow \infty$, and variants of spectral methods \citep{McSherry01,Vu14,LeiZ14} are exactly consistent when $\rho_n n/\log n > C$
for some constant $C$ depending on $\pi_0$, $B_0$ only.
For the simple spectral clustering algorithm used in our numerical study, the exact consistency
result has not been established. In the following, we first establish the approximately consistent recovery result
when $K$ is the true number of blocks.

\begin{theorem}[Consistency of spectral clustering]\label{thm:sbm-consist}
  Under assumptions A1--A2, for each fold split in NCV, the $\hat g$ estimated
   using spectral clustering as described in \Cref{subsec:estimation} is approximately consistent with rate $(n\rho_n)^{-1}$.
\end{theorem}

Now we state the main theorem. For a given $V$-fold block-wise partition of $A$, Let
$$L(A)=\sum_{v=1}^V\sum_{i,j\in\mathcal N_v,i\neq j}\ell(A_{ij},P_{ij})\,.$$

\begin{theorem}\label{thm:error-bound-main}
  Under conditions A1--A3, with the loss function $\ell(a,p)=(a-p)^2$, for a candidate
  $\tilde K$, we have
  \begin{enumerate}
    \item [(a)] When $\tilde K<K $, for \emph{any} estimator $(\hat g, \hat B)$, we have
  $$
  \hat L(A, \tilde{K}) - L(A)
  \ge
    c(n\rho_n)^2 +O_P(1)\,.
  $$
\item  [(b)] When $\tilde{K}= K$, if $\hat g$ is exactly consistent and $\hat B$ is estimated
  as in \eqref{eq:B-est}, then
  $$
  \left|\hat L(A,\tilde{K}) - L(A)\right|=O_P(n\rho_n^{3/2}).
  $$

  \item [(c)] When $\tilde{K}=K$, if $\hat g$ is approximately consistent with rate $\eta_n$ and $\hat B$
  is estimated as in \eqref{eq:B-est}, then
  \begin{align*}
    \left|\hat L(A,\tilde{K})-L(A)\right|=O_P(n\rho_n^{3/2}+n^2\rho_n\eta_n)\,.
  \end{align*}
  \end{enumerate}
\end{theorem}

As a consequence,
when the candidate set $\mathcal K$ is fixed and contains the truth, we have the following
guarantee against under selection.
\begin{corollary}\label{cor:underestimation}
  If the candidate set $\mathcal K$ is fixed, independent of $n$ and $K\in \mathcal K$.
then
  under conditions A1--A2, with the loss function $\ell(a,p)=(a-p)^2$, we have
  \begin{enumerate}
  \item [(a)] When $\hat g$ is exactly consistent and $\hat B$ is estimated
  as in \eqref{eq:B-est}, we have $$\lim_{n\rightarrow\infty}P(\hat K < K)= 0\,$$
  \item [(b)] When the simple spectral clustering method as described in \Cref{subsec:estimation} is used to estimate $\hat g$ and $\hat B$ is estimated as in \eqref{eq:B-est}, if $\rho_n=\omega(n^{-1/2})$ and  $B_0$ has positive smallest singular value, we have
  $$\lim_{n\rightarrow\infty}P(\hat K < K)=0.$$
  \end{enumerate}
\end{corollary}

The proofs of \Cref{thm:sbm-consist} and \Cref{thm:error-bound-main} are given in the \Cref{sec:proof}. Part (a) of \Cref{cor:underestimation} is a direct consequence of part (a), (b) in \Cref{thm:error-bound-main}. Part (b) can be easily derived by combining \Cref{thm:sbm-consist} and part (a), (c) in \Cref{thm:error-bound-main}.

\begin{remark}
  As commonly known for cross-validation methods, there is no
  corresponding theoretical guarantee against overestimation.
Cross-validation methods typically protect against over fitting
by evaluating the fit on an independent subsample. Here we
provide further detail on this idea for the NCV method in stochastic block models.
For $\tilde{K}>K$, the main challenge of analyzing the over fitting is the characterization of
$\hat g$.
Consider two cases. First, if one of the estimated communities $\hat I_l=\{i:\hat g_i=l\}$
contains a substantial proportion of nodes from two true communities, then
the cross-validated predictive loss can be shown to be large using the same argument as in part (a)
of \Cref{thm:error-bound-main}.  Second, if each of the estimated community
contains mostly nodes from one true community, in this case at least one true community is artificially split by the algorithm, and a corresponding $B_{k,k'}$ is estimated separately in more than one of these artificially
split blocks. The difference between these separate estimates
mainly reflects spurious fluctuations due to randomness, which will only
lead to larger predictive loss when evaluated on an independent set of data.
\end{remark}

\section{Degree corrected block models and further extensions}
\label{sec:dcbm}
\subsection{Choosing $K$ for degree corrected block models}
The degree corrected block model \citep{KarrerN11} is a generalization of the stochastic block
model.  Given membership vector $g$ and community-wise connectivity matrix $B$, the presence
of an edge between nodes $i$ and $j$ is represented by a Bernoulli random variable $A_{ij}$
with
\begin{equation}
  P(A_{ij}=1)=1-P(A_{ij}=0)=\psi_i \psi_j B_{g_i g_j}, 
  \label{eq:dcbm}
  \end{equation}
  where $\psi_i>0$ represents the 
\emph{individual activeness} of node $i$.  Thus the degree corrected block model is
parameterized by a triplet $(g, B, \psi)$, with identifiability constraint 
$\max_{i:g_i=k}\psi_i=1$ for all $k=1,...,K$.
The regular stochastic block model is a special case with $\psi_i=1$ for all
$i$.
 Recently, efficient community
recovery methods have been developed for degree corrected block models
 with high accuracy under mild conditions
 \citep[see, for example,][]{ZhaoLZ12,Jin12,ChaudhuriCT12,LeiR14}.  We now extend the procedure described in 
 \Cref{sec:sbm-ncv} to degree corrected block models. Algorithm 1 is general enough to cover degree-corrected block model. We only need to modify the parameter estimation step. For estimating $(\hat g, \hat B, \hat \psi)$, we consider a
spherical spectral clustering method.

{\bf Spherical spectral clustering:}
\begin{enumerate}[itemsep=-3pt,topsep=0pt]
  \item [] Input: Rectangular $n_1\times n$ matrix $A^{(1)}$, a candidate number of communities $\tilde K$.
  \item Let $\hat U$ be the $n\times K$ matrix consisting of the top $\tilde K$ right singular
  vectors of $A^{(1)}$.
  \item Let $\tilde U$ be the matrix obtained by scaling each row of $\hat U$ to unit norm.
  \item Output $\hat g$ by applying the $k$-median clustering 
  algorithm to the rows of $\tilde U$.
\end{enumerate}

The normalization step in the spherical spectral clustering algorithm 
decouples the effect of node activeness $\psi$
from the community structure.
As shown in the proof of \Cref{thm:dcbm} below, the community information is contained
in the normalized matrix $\tilde U$, whereas the node activeness information
is contained in the row norms of $\hat U$. 

The community recovery is obtained by
a $k$-median clustering algorithm, which finds a collection of center points to
minimize
the sum of $\ell_2$ distance from each data point to its nearest center, 
instead of the squared $\ell_2$ distance as in the $k$-means.  To be precise,
given input matrix $\tilde U$ and number of centers $\tilde K$, the $k$-median clustering solves, possibly with approximation, the following optimization
problem:
$$
\min_{u_1,...,u_{\tilde K}\in \mathbb R^{\tilde K},\ g\in\{1,...,\tilde K\}^n}
\sum_{i=1}^n \|\tilde u_{i}-u_{g_i}\|\,,
$$
where $\tilde u_{i}$ is the $i$th row of $\tilde U$.  Approximate solutions 
within a constant factor from the global optimum
can be found using efficient algorithms \citep{CharikarGTS99,LiS13}.
Our theoretical analysis is applicable to such approximate solutions.
If the matrix $\hat U$ has zero rows, one can apply
the spherical clustering algorithm on the non-zero rows and assign arbitrary
membership to the  zero rows.
Our theoretical analysis shows that with high probability 
the number of zero rows in $\hat U$
is negligible under mild conditions.

To estimate the node activeness parameter $\psi$,
let
\begin{equation}\label{eq:est-psi}
\hat \psi_i' = \text{$\ell_2$~norm~of~the~$i$th~row~of~$\hat U$}\,,
\end{equation}
and $\psi'=(\psi_i':1\le i\le n)$ with
$$\psi_i'=
\frac{\psi_i}{\sqrt{\sum_{j:g_j=g_i}\psi_j^2}}
$$
be the community-normalized version of $\psi$. We will show, in the proof of \Cref{thm:dcbm} below,
 that $\hat\psi'$ is a good estimate of $\psi'$ under appropriate conditions.  Due to the
 scaling identifiability of $\psi$ and $B$, having a good estimate of $\psi'$ is sufficient for our purpose and
 one can proceed with the plug-in estimator:
 \begin{equation}\label{eq:B'}
   \hat B_{k,k'}'=\left\{
   \begin{array}{ll}
   {\displaystyle \frac{\sum_{i\in \mathcal N_{1,k}, j\in\mathcal N_{1,k'}\cup \mathcal N_{2,k'}}A_{ij}}
    {\sum_{i\in \mathcal N_{1,k}, j\in\mathcal N_{1,k'}\cup 
    \mathcal N_{2,k'}}\hat\psi'_i\hat\psi'_j}}\,, & k\neq k\,,\\
    &\\
    {\displaystyle\frac{\sum_{i,j\in \mathcal N_{1,k}, i<j} A_{ij} +
    \sum_{i\in\mathcal N_{1,k}, j\in\mathcal N_{2,k}}A_{ij}}{
    \sum_{i,j\in \mathcal N_{1,k}, i<j} \hat\psi_i'\hat\psi_j' +
        \sum_{i\in\mathcal N_{1,k}, j\in\mathcal N_{2,k}}\hat\psi_i'\hat\psi_j'}}\,,& k=k'\,.
    \end{array}
   \right.
 \end{equation}
The estimated $P_{ij}=E(A_{ij})$ to be used for validation is then
$$
\hat P_{ij} = \hat\psi'_i\hat\psi_j' \hat B'_{\hat g_i,\hat g_j}\,.
$$

To investigate theoretical properties of these estimators, we assume that there are no 
overly inactive nodes.
\begin{enumerate}
  \item [(A3)] $\inf_{1\le i\le n}\psi_i\ge \psi_0$ for a positive constant $\psi_0$.
\end{enumerate}

The performance analysis of NCV for DCBM's is beyond the scope of this paper.
Here we provide accuracy guarantee for community recovery and the edge probability estimation.

\begin{theorem}\label{thm:dcbm}
  Under (A1)--(A3), assume $(\hat g, \hat B, \hat \psi)$ is obtained by spherical spectral clustering combined with Equations \eqref{eq:est-psi} and \eqref{eq:B'}, 
  then for each fold split we have
  \begin{enumerate}
  \item [(a)] if $\alpha_n\geq c \log n/n$ for a large enough $c$, then with probability tending to one $\hat g$ agrees with $g$ on all but $O(\sqrt{n/\alpha_n})$ nodes; 
  \item [(b)] if $\alpha_n^{-1}=o(n^{1/3})$, then
    $\hat P_{ij}=P_{ij}(1+o_p(1))$ for all but a vanishing proportion of node pairs.
  \end{enumerate}
    \end{theorem}
\Cref{thm:dcbm} is proved in \Cref{sec:proof-dcbm}. Part (a) establishes approximate consistency of spherical spectral clustering applied on the rectangle fitting set of node pairs.  Part (b)
requires a larger average edge probability so that the estimation error of $\hat B$ is well-controlled.  In this case, part (a) of the theorem suggest that 
the proportion of mis-clustered nodes is $o(n^{-1/3})$.

\subsection{Choosing model types and $K$ simultaneously}
The above extension to degree corrected block models allows us to compare and choose, 
for a given
adjacency matrix, between the regular stochastic block model and the degree corrected block 
model.
Sometimes it is desirable to tell if the degree heterogeneity in an observed network
can be explained by pure random fluctuation in a stochastic block model
 \citep[see, for 
example,][]{Yan14}.

Our V-fold NCV can be used simultaneously to choose between
the regular stochastic block model and the degree corrected block model, and to determine
the number of blocks. To this end, one just needs to calculate the regular
stochastic block model validation error $\hat L_{\rm sbm}(A,\tilde K)$,
 and the degree corrected block model validation error
$\hat L_{\rm dcbm}(A,\tilde K)$, for a collection of values of $\tilde K$ as described in
\Cref{subsec:ncv}.  The best model is chosen
 by finding the overall smallest cross-validation loss.  We illustrate this method
 on simulated data and on a political blog data in \Cref{sec:experiment}.

\section{Numerical Experiments}
\label{sec:experiment}
In this section, we illustrate the performance of our proposed NCV method by three simulations and one data example.

\noindent{\bf Simulation 1: edge sparsity and community imbalance.}
This simulation is designed to investigate the performance of choosing $K$ for
stochastic block models under different levels of edge sparsity and community
size imbalance.
We use the community-wise edge probability matrix $B = rB_0$, where the diagonal
entries of $B_0$ are 3 and
off-diagonal entries are 1.
The sparsity level is controlled by $r\in(0,1/3)$.
We use a sequence of $r \in \{ 0.01, 0.02, 0.05, 0.1, 0.2\}$, so
that for $n=1000$ the smallest expected degree ranges from 12 to 400.
Let $n_1$ be the
size of the smallest community,
and the size of each of the remaining $K-1$ communities be $(n-n_1)/(K-1)$. We
generate edges according to the stochastic block model \Cref{eq:sbm}.
For each combination of $(r,K,n_1)$, three-fold NCV model selection is carried out for
50 independently drawn adjacency matrices.
\Cref{fig:simu1} shows the proportion of correct model selection among these 50
repetitions as
functions of $r$ for different $n_1$ and $K = 2, 3, 4$. As expected, the
performance is better as $r$ and $n_1$ increase. In particular, for $K=2$,
in the most balanced case where $n_1 = 500$, the proposed NCV can perfectly
choose the true number of clusters even for the sparsest case where $r = 0.01$, whereas in
the most imbalanced case where $n_1 = 100$, there is a phase transition near $r = 0.1$. The
curve for $n_1 = 200$ is in between. The same phenomenon is observed for $K=3$ and
$K=4$. The proposed NCV can almost perfectly pick out $K$ for relatively balanced
community sizes, even for very sparse cases. For imbalanced cases, one needs to have
moderate expected degrees for the nodes in the smallest community. We note that
community recovery for a given $K$ is an integrated step in the proposed NCV method, so
it is expected that the performance of NCV is closely related to the difficulty of the
community recovery problem when knowing the true $K$, which may depend on the
particular community recovery method used in NCV.

\begin{figure}[t]
\centering
\makebox{\includegraphics[scale=.31]{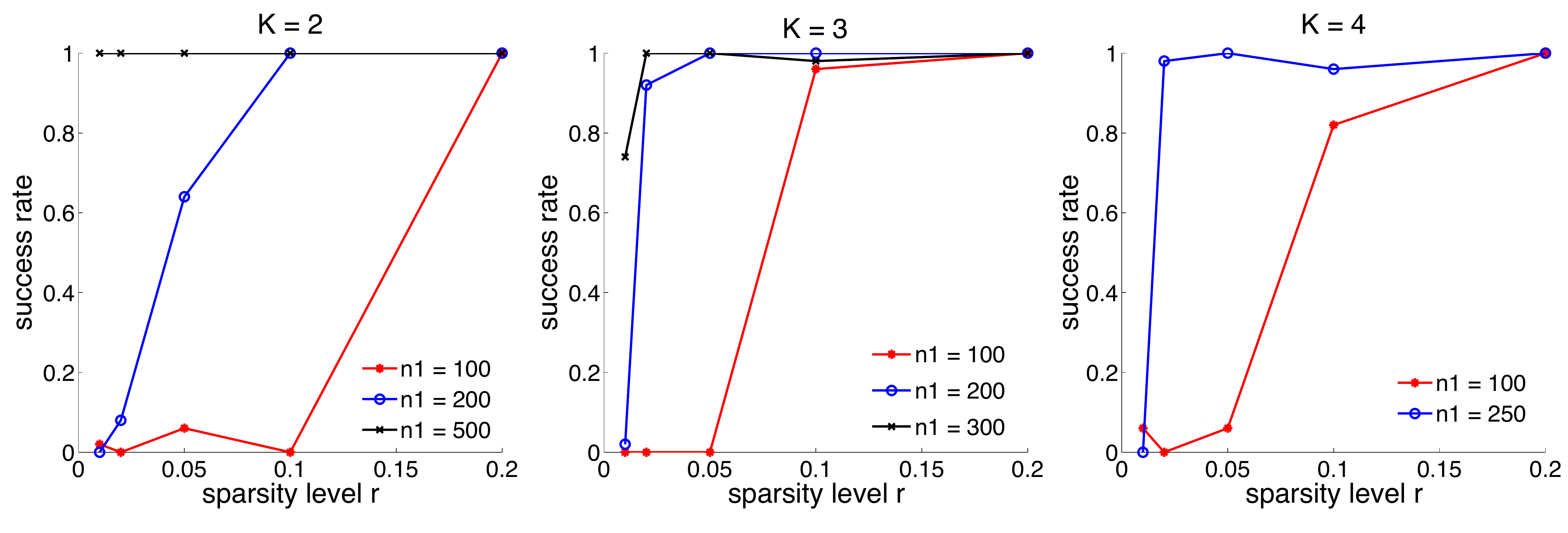}}
\caption{\label{fig:simu1}Results for {\bf Simulation 1}: reporting the proportion of
correct estimate
of $K$ for stochastic block models, for $K= 2 ,3, 4$, under various
sparsity levels $r \in \{0.01, 0.02,
0.05, 0.1, 0.2\}$, and various sizes of the first community $n_1$. The
number of nodes
is $1000$.
}
\end{figure}

\noindent{\bf Simulation 2: general block structures and comparison to recursive
bipartition.} This simulation is designed to further investigate the
proposed NCV method under general block structures of networks,
and meanwhile to compare the
proposed NCV method with the recursive testing procedure proposed in
\cite{BickelS13test}. We generate symmetric $B$ randomly as follows. For each
upper-triangle entry of $B$, we generate a random number from ${\rm Unif}(0, 0.5)$. The
upper bound $0.5$ is set to exclude unrealistically dense networks that are of less
interest. We only use $B$ matrices whose $K$th singular values are in the upper three
quarters
and therefore have relatively well-formed $K$-block structures. The membership vector
$g$ is generated from multinomial distribution $(n,\pi)$ with equal probability $\pi
= (1/K,\dots,1/K)$. For each simulated data, we applied three-fold NCV method as well
as the recursive bipartition algorithm developed in \cite{BickelS13test} with $\alpha
= 0.01$. The basic idea of the recursive bipartition method is to divide the nodes
into two clusters if $K=1$ is rejected at level $\alpha$, and then recursively test
$K=1$ vs $K>1$ on each of the two sub-networks until failing to reject $K=1$.  The
success rates in 50 simulations for each combination of $n = 600, 1200$ and $K = 1,
2, 3,
4$ are shown in \Cref{fig:simu2}. As expected, both methods benefit from a larger
sample size (top row vs
bottom row). The task of determining $K$ gets harder as the true number of communities
gets larger (from left column to right column).
The proposed NCV method performs uniformly better
than the bipartition method.  The NCV method is also much faster than
the bipartition method, where the latter requires a small bootstrap sample to adjust the null
distribution at each testing step.  The simulation design
suggests that the proposed NCV has very
satisfactory performance
under very general structures of $B$. For $n=1200$, the empirical success rate
of NCV
achieves 100\% for $K=1,\ 2$, 84\% for $K=3$, and 72\% for $K=4$.

\begin{figure}[t]
\centering
\makebox{\includegraphics[scale=0.6]{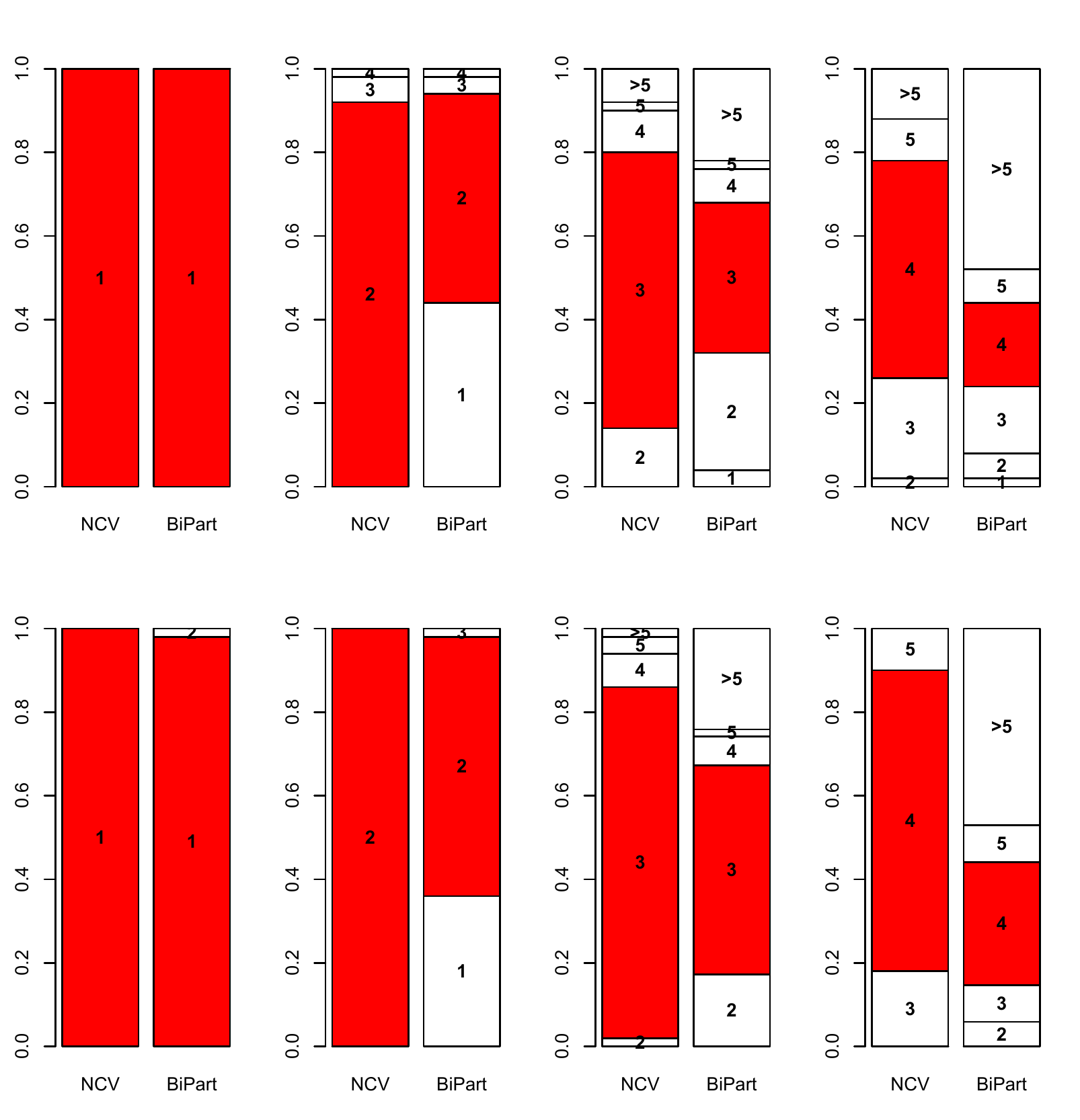}}
\caption{\label{fig:simu2}Results for {\bf Simulation 2}: reporting the proportion of
selected $K$ by NCV (the
proposed method) and BiPart (\cite{BickelS13test}), for true $K= 1, 2 ,3, 4$ (from left
to right), and sample size $n = 600$ (top row) and $n= 1200$ (bottom row).
}
\end{figure}

\noindent{\bf Simulation 3: degree corrected block models.} This simulation is designed
to demonstrate the performance of selecting between the stochastic block model and the
degree-corrected block model with simultaneous selection of $K$. We use a $B$ matrix
whose diagonal is 0.25 and off-diagonal is 0.1, which gives a moderate sparsity level
for stochastic block models. For degree-corrected block model, the degree parameter
$\psi$ is generated from ${\rm Unif}(0.2, 1)$, and normalized to have block-wise
maximum value 1. The edges are generated according to \Cref{eq:dcbm}. The network is
much sparser in presence of the degree parameter $\psi$ and the inference
problem is harder.
Three-fold NCV is used to simultaneously choose the model type $T$ from
$T=\text{``SBM''}$ or $T=\text{``DCBM''}$,
 and the number of
communities
$K$. \Cref{tab:simu3} shows the proportion of correct model type selection $\hat T=T$
and  proportion of correct choice of $K$ given correct model type selection.
Data are generated 50 times from both the
stochastic block model and the degree corrected block model,
for each combination of $K = 1, 2, 3, 4$ and
$n = 300, 600, 1200$. We observe that when the true model type
is stochastic block model, NCV can almost perfectly pick out the correct model
and correct $K$ for various combinations of $K$ and $n$. As expected, a relatively
larger
sample size is needed to get good performance when the network is generated from a
degree corrected block model. Our simulation shows that for $n=1200$, NCV can
almost always pick out the correct DCBM model with the right $K$.

\begin{table}
\caption{\label{tab:simu3}Results for {\bf Simulation 3}: proportion of selecting
the correct model type, and choosing the correct $K$ given correct model type
selection, from 50 independent simulations. The true models are generated from
stochastic block models (SBM)
or degree corrected block models (DCBM), for true $K = 1, 2, 3, 4$ and $n = 300, 600,
1200$. }
\centering
\fbox{%
\begin{tabularx}{0.96\textwidth}{cccccccccc}
 \multicolumn{2}{c}{}&\multicolumn{4}{c}{\it SBM }&\multicolumn{4}{c}{\it DCBM}\\
\multicolumn{2}{c}{}& $K =1$ & $2$ & $3$ & $4$ & $K=1$ & $2$ & $3$ &$4$ \\
\hline
$n = 300$ & $\hat P(\hat T=T)$ & 1 & 1& 1& 1 &1 &0.68 & 0.44 & 0.42  \\
  & $\hat P(\hat K=K|\hat T=T)$ & 1 & 1& 0.98& 0.92& 1 & 0.41 & 0 & 0 \\
\hline
$n = 600$ & $\hat P(\hat T=T)$ & 1 & 1& 1& 1 &1 &1 & 0.96 & 0.98  \\
  & $\hat P(\hat K=K|\hat T=T)$ & 1 & 1& 1& 0.98& 1 & 1 & 0.42 & 0 \\
 \hline
 $n= 1200$& $\hat P(\hat T=T)$ & 1 & 1& 1& 1 &1 &1 & 1 & 1 \\
  & $\hat P(\hat K=K|\hat T=T)$ & 1 & 1& 1& 0.98 &1 &1 & 1 & 1\\
\end{tabularx}}
\end{table}

\noindent\textbf{Data example: political weblogs.} The political blog data was collected
and analyzed in \cite{PolBlog}.  The data set contains snapshots of over one thousand
weblogs shortly before the 2004 U.S. Presidential Election, where the nodes are weblogs, and
edges are hyperlinks.  The nodes are labeled as being either liberal or conservative, which can be treated as two well-defined communities. The degree corrected block model is believed to fit better than the stochastic block model to this data with two communities \citep{KarrerN11,ZhaoLZ12,Jin12}.  To illustrate the NCV method
for simultaneously choosing between the regular stochastic block model and the degree
corrected block model, and choosing the number of communities $K$, we apply
three-fold NCV
to the largest connected component in the network which contains 1222 nodes.  The NCV method consistently selects the degree corrected block model
with two communities.  The cross-validated negative log-likelihood for all candidate models
is plotted in \Cref{fig:PolBlog} for a typical block splitting.
We repeated the NCV selection 100 times using independent random
block splittings.
The NCV method selected DCBM and $K=2$ in 99 out of 100 repetitions,
where the one failure was due to non-convergence of
$k$-means in spectral clustering.

\begin{figure}[t]
\centering
\makebox{\includegraphics[scale=0.7]{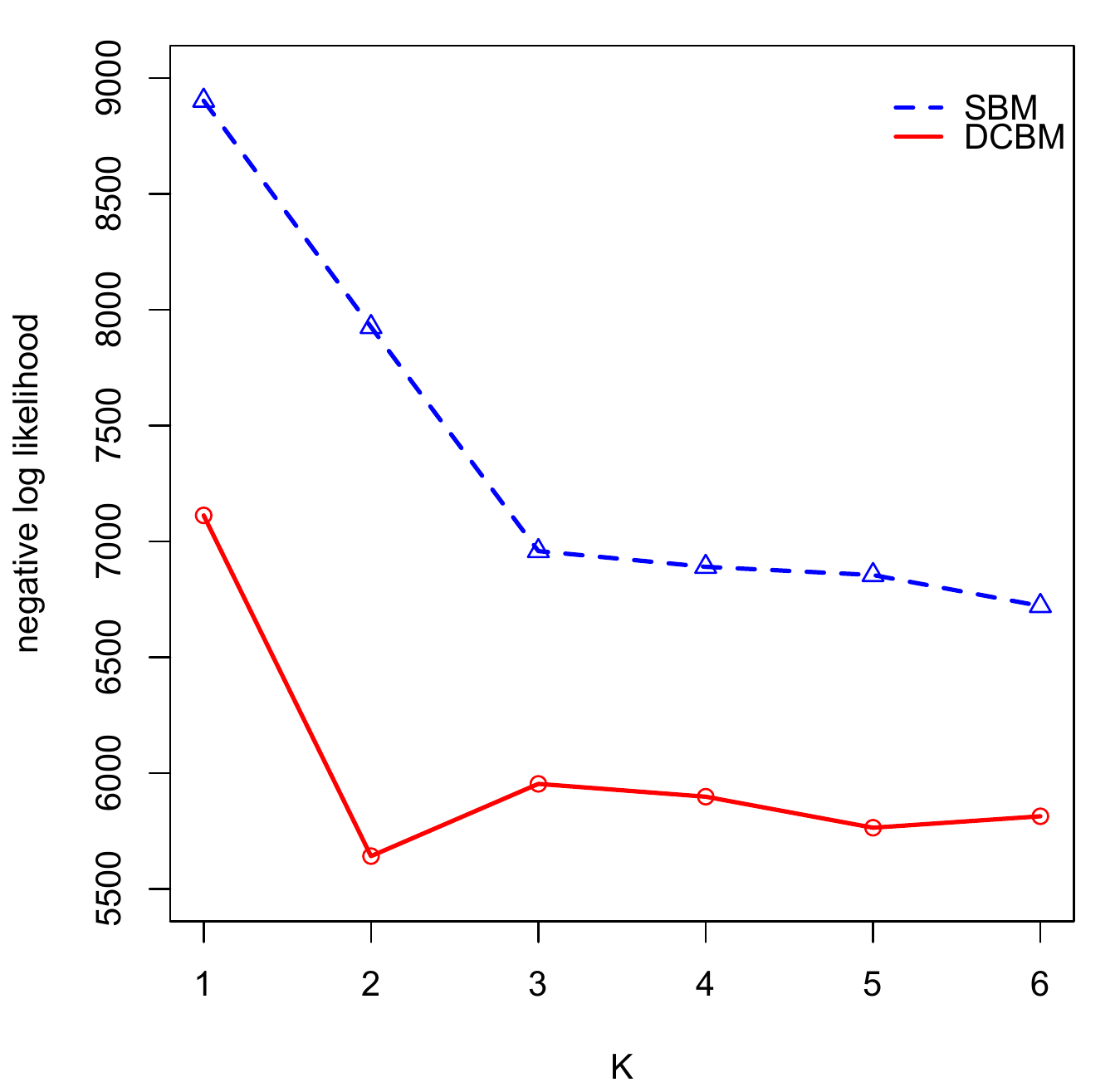}}
\caption{\label{fig:PolBlog}Results for {\bf the political blogs data}: reporting the
three-fold cross-validated negative
log-likelihood of all candidate models from one random block splitting.
Dashed line: stochastic block models;
solid line: degree corrected block models. The results are consistent over
100 repeated random block splittings.
}
\end{figure}

\section{Discussion}
\label{sec:discussion}
\paragraph{Further extensions}
In general, the network cross-validation approach proposed in this paper
is applicable
to network models where (i) edges form independently given an appropriate set of model 
parameters; 
and (ii) the edge probabilities can
be estimated accurately using a subset of rows of the adjacency matrix.
The stochastic block model and the degree corrected model are good examples satisfying these two properties.  There are
other popular network models in this category, such as the random dot-product graph.
 The random dot-product graph model 
\citep{Young07} assumes that each node $i$ has an embedding $v_i$ on a subset
of the $d$-dimensional
unit sphere, and that given the embedding
the edge between node $i$ and node $j$ is an independent Bernoulli random
variable with parameter $\langle v_i,v_j\rangle$.
This is a special case of the latent space model \citep{Hoff02}.
The latent vectors can be accurately estimated using spectral methods \citep{Sussman13},
which can be adapted naturally so that the model parameters can be estimated using only the fitting subset of
node pairs in $A$.

\paragraph{Effect of the number of folds}
In general, cross validation methods are insensitive to the number of folds.
The same intuition empirically holds true for the proposed NCV method.  
However, there is
slight difference between the NCV framework and the traditional cross-validation.
Unlike traditional cross validation where each data point is
included in a testing sample, In V-fold NCV only the diagonal
blocks are used as testing samples and hence the proportion of edges included in testing
samples is roughly $1/V$.
On the other hand, the ratio between the sizes of fitting and
testing samples in a single fold is $(V^2 - 1)$ to 1 for NCV, and $(V-1)$ to 1 for 
traditional cross-validation. 
Roughly speaking,
having a larger value of $V$ will rapidly increase the estimation accuracy in the 
fitting stage
but will reduce the testing sample size.  
In our numerical experiments we found $V=3$ a reasonable choice for
most cases, which is roughly comparable to a 9-fold traditional cross validation in 
terms of the fitting and testing sample size ratio.  

\appendix
\section{Proofs}\label{sec:proof}

\subsection{Proof of Theorem 2}
\begin{proof}
We focus on a single fold in NCV. The claimed results follow by summing over all folds.  
Again let $\mathcal N_1$ be the nodes corresponding to the fitting part and
$\mathcal N_2$ be those in the testing part.
Let $I_{l}=\{i\in\mathcal N_2:g_i=l,\}$ ($1\le l\le K$) and $\hat I_{k}= \{i\in \mathcal N_2:\hat g_i=k\}$ ($1\le k\le \tilde{K}$).

\textbf{Case 1: $\tilde{K}<K$}.

According to tail probability inequalities for hypergeometric distributions (see \Cref{lem:split}), with overwhelming probability over the random data splitting, the testing block $A^{(vv)}$ is also a realization of a $K$-block SBM satisfying assumptions A1--A2,
with a different constant $\pi_0'$.

When $\tilde{K}<K$,  there exist $1\le k\le \tilde{K}$, $1\le l_1<l_2\le K$ such that
  $|\hat I_k \cap I_{l_j}|\ge |I_{l_j}|/\tilde{K}\ge \pi_0' n/|\tilde{K}|$ for $j=1,2$.
  Because $B_0$ does not have identical rows, there exists $1\le l_3\le K$ such that
  $B_0(l_1, l_3)\neq B_0(l_2,l_3)$. There exists a $k'$ such that $|\hat I_{k'}\cap I_{l_3}|\ge
  |I_{l_3}|/\tilde{K}\ge \pi_0' n/\tilde{K}$.
  We focus on the case $k\neq k'$ and $l_1,l_2,l_3$ are distinct.  The other cases, such as $k=k'$
  or $l_1=l_3$, or both, can be dealt with similarly.
  Without loss of generality, assume $k=1,k'=2$, $l_1=1,l_2=2,l_3=3$.

    Let
$\mathcal T_{k,k',l,l'}$  be the set of unique pairs $(i,j)$ such that
$i\in \hat I_{k} \cap I_l$, $j\in \hat I_{k'}\cap I_l$, and $i\neq j$ (that is, $(i,j)$ and
$(j,i)$ shall be counted as one unique pair in case of $k=k'$ and $l=l'$).
Let $\hat p$ be the average of $A_{ij}$ over $(i,j)\in \mathcal T_{1,2,1,3}\cup\mathcal T_{1,2,2,3}$
and $\hat p_{k,k',l,l'}$ be the average of $A_{ij}$ over $(i,j)\in\mathcal T_{k,k',l,l'}$.
Then, letting $L^{(v)}(A)$ be the terms in $L(A)$ corresponding to the $v$th fold, \begin{align*}
  &\hat L^{(v)}(A,\tilde{K})-L^{(v)}(A)\\=&\sum_{k,k',l,l'}\sum_{(i,j)\in \mathcal T_{k,k',l,l'}}\left[(A_{ij}-\hat P_{ij})^2-(A_{ij}-P_{ij})^2\right]\\
  =&\sum_{k,k',l,l'}\sum_{(i,j)\in \mathcal T_{k,k',l,l'}}\left[(A_{ij}-\hat B_{k,k'})^2-(A_{ij}-B_{l,l'})^2\right]\\
  =&\sum_{(i,j)\in \mathcal T_{1,2,1,3}}\left[(A_{ij}-\hat B_{1,2})^2-(A_{ij}-B_{1,3})^2\right]
  +\sum_{(i,j)\in \mathcal T_{1,2,2,3}}\left[(A_{ij}-\hat B_{1,2})^2-(A_{ij}-B_{2,3})^2\right]\\
  &+\quad
  \sum_{(k,k',l,l')\notin \{(1,2,1,3),(1,2,2,3)\}}
  \sum_{(i,j)\in \mathcal T_{k,k',l,l'}}\left[(A_{ij}-\hat B_{k,k'})^2-(A_{ij}-B_{l,l'})^2\right]\\
  \ge & \sum_{(i,j)\in \mathcal T_{1,2,1,3}}
  \left[(A_{ij}-\hat p)^2-(A_{ij}-B_{1,3})^2\right]
    +\sum_{(i,j)\in \mathcal T_{1,2,2,3}}\left[(A_{ij}-\hat p)^2-(A_{ij}-B_{2,3})^2\right]\\
    &+\quad
    \sum_{(k,k',l,l')\notin \{(1,2,1,3),(1,2,2,3)\}}
    \sum_{(i,j)\in \mathcal T_{k,k',l,l'}}\left[(A_{ij}-\hat p_{k,k',l,l'})^2-(A_{ij}-B_{l,l'})^2\right]\\
    = & I + II + III.
\end{align*}
Let $\lambda=|\mathcal T_{1,2,1,3}|/(|\mathcal T_{1,2,1,3}+\mathcal T_{1,2,2,3}|)$.
Then $\pi_0'^2/(\tilde{K}^2+\pi_0'^2)\le \lambda\le \tilde{K}^2/(\tilde{K}^2+\pi_0'^2)$,
and $\hat p = \lambda \hat p_{1,2,1,3}+(1-\lambda)\hat p_{1,2,2,3}$.
If $|\mathcal T_{k,k',l,l'}|>0$, using Bernstein's inequality we have
\begin{align}\label{eq:bern}
|\hat p_{k,k',l,l'} - B_{l,l'}|=O_P\left(\sqrt{\frac{\rho_n}{|\mathcal T_{k,k',l,l'}|}}\right)\,.
\end{align}
\begin{align*}
  I = & |\mathcal T_{1,2,1,3}|\left[(\hat p - \hat p_{1,2,1,3})^2 - (\hat p_{1,2,1,3}-B_{1,3})^2\right]\\
  =& |\mathcal T_{1,2,1,3}|\left[(1-\lambda)^2(\hat p_{1,2,1,3}-\hat p_{1,2,2,3})^2
  -(\hat p_{1,2,1,3}-B_{1,3})^2\right]\\
  \ge & |\mathcal T_{1,2,1,3}|\bigg[\frac{(1-\lambda)^2}{2}( B_{1,3}- B_{2,3})^2
  - (1-\lambda)^2\left((\hat p_{1,2,1,3}-B_{1,3}) - (\hat p_{1,2,2,3}-B_{2,3})\right)^2\\
  &\qquad\qquad
  -(\hat p_{1,2,1,3}-B_{1,3})^2\bigg]\\
  \ge & |\mathcal T_{1,2,1,3}|\bigg[\frac{(1-\lambda)^2}{2}( B_{1,3}- B_{2,3})^2
  - 2(1-\lambda)^2\left(\hat p_{1,2,1,3}-B_{1,3}\right)^2 \\
  &\qquad\qquad - 2(1-\lambda)^2\left(\hat p_{1,2,2,3}-B_{2,3}\right)^2
  -(\hat p_{1,2,1,3}-B_{1,3})^2\bigg]\\
  \ge & c (n\rho_n)^2 + O_P(\rho_n)\,.
\end{align*}
where the constant $c$ depending on $\pi_0$, $B_0$ and $\tilde{K}$ only,
the first term in the last inequality comes from the fact that $|\mathcal T_{1,2,1,3}|\ge \pi_0^2n^2/\tilde{K}^2$ and $|B_{1,2}-B_{1,3}|\ge c' \rho_n$ for some $c'$ depending only on $B_0$,
the second term comes from \eqref{eq:bern} and the fact that $|T_{1,2,1,3}|\asymp |T_{1,2,2,3}|$.

Similarly,
$$
II\ge  c (n\rho_n)^2 + O_P(\rho_n)\,.
$$
and
\begin{align*}
  III \ge -\sum_{(k,k',l,l')\notin \{(1,2,1,3),(1,2,2,3)\}}
    |\mathcal T_{k,k',l,l'}|(\hat p_{k,k',l,l'}-B_{l,l'})^2 = O_P(\rho_n)\,.
\end{align*}

\textbf{Case 2: $\tilde{K}=K$ with exactly consistent recovery}\\
In this case we focus on the event $\hat g = g$, which has $1-o(1)$ probability.
Then $\epsilon_n\coloneqq\sup_{1\le l,l'\le K}|\hat B_{l,l'}-B_{l,l'}|=O_P(\sqrt{\rho_n}/n)$
by Bernstein's inequality and
$\hat B_{l,l'}$ are independent with $A_{ij}$ for $(i,j)$ in the testing set.

For $1\le k,k'\le K$, let $\mathcal T_{l,l'}$ be the collection of pairs $(i,j)$ in the testing set such that
$i\in I_k$, $j\in I_{k'}$ and $i\neq j$.
\begin{align*}
  \left|\hat L^{(v)}(A,\tilde{K})-L^{(v)}(A)\right| \le & \sum_{k,k'} \sum_{(i,j)\in\mathcal T_{k,k'}}
  \left|(A_{ij}-\hat B_{k,k'})^2 - (A_{ij}-B_{k,k'})^2\right|\\
  = & \sum_{k,k'} \sum_{(i,j)\in\mathcal T_{k,k'}}
    \left|-2A_{ij}(\hat B_{k,k'}-B_{k,k'})+(\hat B_{k,k'}+B_{k,k'})(\hat B_{k,k'}-B_{k,k'})\right|\\
  \le &
  \sum_{k,k'} \sum_{(i,j)\in\mathcal T_{k,k'}}
    \left[2\epsilon_n A_{ij}+(2\rho_n+\epsilon_n)\epsilon_n\right]\\
    \le & O_p(n^2\rho_n)\epsilon_n+n^2(2\rho_n+\epsilon_n)\epsilon_n=O_P(n\rho_n^{3/2})\,.
\end{align*}
Because $n\rho_n^{3/2}=o((n\rho_n)^2)$, we have
\begin{align*}
P(\hat K< K)\le &\sum_{\tilde{K}< K}P\left[\hat L(A,\tilde{K})<\hat L(A,K)\right]  \\
\le &\sum_{\tilde{K}<K} \sum_{1\le v\le V}P\left[\hat L^{(v)}( A,\tilde{K})<\hat L^{(v)}(A,K)\right]
=o(1)\,.
\end{align*}
This proves part (a) of \Cref{cor:underestimation}.

\textbf{Case 3: $\tilde{K}=K$ with approximately consistent recovery}

  Without loss of generality we assume the optimal permutation to match $\hat g$
  and $g$ is the identity.
  Let $\hat{\mathcal U}_{k,k'}$ be the set of distinct pairs $(i,j)$ in the fitting
  subset such that $\hat g_i=k$, $\hat g_i=k'$, and $i\neq j$.
  Let $\mathcal U_{k,k'}$ be the set of distinct pairs $(i,j)$ in the fitting
  subset such that $g_i=k$, $g_i=k'$, and $i\neq j$.
  Let $S=\bigcup_{k=1}^K \hat I_{k}\Delta I_{k}$, $\mathcal S=\bigcup_{k,k'}\hat{\mathcal U}_{k,k'}\Delta \mathcal U_{k,k'}$.
  On the event $|S|\le \eta_n n$, under Assumptions (A1) and (A2) we
  have $|\mathcal S|= O(\eta_n n^2)$, and hence
  $|\hat{\mathcal U}_{k,k'}|=|\mathcal U_{k,k'}|(1+O(\eta_n))$ for all $k,k'$.
  \begin{align*}
    &\left|\hat B_{k,k'}-B_{k,k'}\right|\\=&\left|
    \frac{1}{|\hat {\mathcal U}_k|}\left(\sum_{(i,j)\in \mathcal U_{k,k'}} A_{ij}\right)-B_{k,k'}\right|+\frac{1}{|\hat{\mathcal U}_{k,k'}|}\sum_{(i,j)\in
    \hat{\mathcal U}_{k,k'}\Delta \mathcal U_{k,k'}} A_{ij}\\
    \le& \frac{|\mathcal U_{k,k'}|}{|\hat{\mathcal U}_{k,k'}|}
    \left|\frac{1}{|\mathcal U_{k,k'}|}\sum_{(i,j)\in\mathcal U_{k,k'}}A_{ij}-B_{k,k'}\right|
    +\left|\frac{|\mathcal U_{k,k'}|}{|\hat{\mathcal U}_{k,k'}|}-1\right|B_{k,k'}
    +\frac{1}{|\hat{\mathcal U}_{k,k'}|}\sum_{(i,j)\in
    \hat{\mathcal U}_{k,k'}\Delta \mathcal U_{k,k'}} A_{ij}\\
    \le& (1+O(\eta_n)) O_P\left(\sqrt{B_{k,k'}/|\mathcal U_{k,k'}|}\right)
    +O(\eta_n)B_{k,k'}+(1+O(\eta_n))\eta_n\\
    = &O_P(\rho_n^{1/2}n^{-1}+\eta_n)\,.
  \end{align*}
  Focusing on the $v$th fold in cross-validation, we similarly have, letting $\epsilon_n=\sup_{k,k'}|\hat B_{k,k'}-B_{k,k'}|$,
\begin{align*}
  \left|\hat L^{(v)}(A,\tilde{K})-L^{(v)}(A)\right|
    \le & O_p(n^2(\rho_n+\epsilon_n)\epsilon_n)=O_p(n\rho_n^{3/2}+n^2\rho_n\eta_n)\,.\qedhere
\end{align*}
\end{proof}

\subsection{Proof of \Cref{thm:sbm-consist} and \Cref{thm:dcbm}}
Next we prove \Cref{thm:sbm-consist} and \Cref{thm:dcbm}.
Let $P$ be the $n\times n$ matrix such that $P_{ij}=B_{g_i g_j}$.
For a particular fold split in NCV,
let $A^{(1)}$ and $P^{(1)}$ be the testing rectangular submatrix of $A$ and $P$ respectively.
Let $\mathcal N_{j,k}^*$ be the nodes in subsample $\mathcal N_j$ belonging to
community $k$ and $n_{j,k}^*=|\mathcal N_{j,k}^*|$ ($j=1,2$, $k=1,...,K$).
For any matrix $M$, let $\sigma_K(M)$ be its $K$th largest singular value.
In the statement of results and the proof, constants $c$, $C$ may take different
values from line to line. We let $\|M\|=\sigma_1(M)$ be the spectral norm of $M$
and $\|M\|_F=(\sigma_1^2(M)+\sigma_2^2(M)+...)^{1/2}$ be the Frobenious norm. 

\subsubsection{Preliminary results}
The following technical lemmas are needed in the proof of \Cref{thm:sbm-consist} and \Cref{thm:dcbm}. 

\begin{lemma}[Size of split community]\label{lem:split}
Under Assumption (A1),
for $n$ large enough we have $\min_{k}n_{1,k}^*\ge \pi_0 n/(2V)$,
with probability at least $1-n^{-1}/2$.
\end{lemma}
The proof of this lemma follows from a simple application of
large deviation bounds for hypergeometric random variables \citep{Skala13} combined with
union bound and is omitted.

\begin{lemma}[Spectral norm error of partial adjacency 
  matrix]\label{lem:spec-bound}
Let $A$ be the adjacency matrix generated from a degree corrected block
model satisfying Assumption A2, with $\rho_n\ge c\log n/n$ for a 
positive constant $c$. Let $\tilde A$ be an arbitrary subset of rows 
of $A$ and $\tilde P$ be the corresponding submatrix of $P$.
We have, for some constant $C$,
  $$
  P\left(\|\tilde A-\tilde P\|\le C\sqrt{n\rho_n}\right)\ge 
  1-n^{-1}/2\,.
  $$
\end{lemma}
\begin{proof}
Observe that $\|\tilde A-\tilde P\|\le \|A-P\|$.  The claimed result follows easily from 
Theorem 5.2 of \cite{LeiR14}, where it has been shown that
  with high probability $\|A-P\|\le C\sqrt{n\rho_n}$.
\end{proof}
\Cref{lem:spec-bound} covers both regular SBM and DCBM. It implies that $\|A^{(1)}-P^{(1)}\|\le C\sqrt{n\rho_n}$ with high probability.

\begin{lemma}[Singular subspace error bound]\label{lem:svd-error}
  Let $\hat M$, $M$ be two matrices of same dimension, and
   $\hat U$ and $U$ be $n\times K$ orthonormal matrices corresponding to the
  top $K$ right singular vectors of $\hat M$ and $M$, respectively.
  Then there exists a $K\times K$ orthogonal matrix $Q$ such that
  $$
  \|\hat U - UQ\|\le \frac{2\sqrt{2}\|\hat M - M\|}{\sigma_K(M)}\,.
  $$
\end{lemma}
\begin{proof}
  If $\|\hat M  - M\|\le \sigma_K(M)/2$, then using Wedin $\sin\Theta$ theorem
  \citep{Wedin72} and Weyl's inequality there exists an orthogonal $Q$ such that
  $\|\hat U- UQ\|\le \|\hat M - M\|/(\sigma_K(M)-\|\hat M - M\|)\le 2\|\hat M-M\|/\sigma_K(M)$.
  If $\|\hat M - M\|\ge \sigma_K(M)/2$, then
  $\|\hat U - UQ\|\le 1 \le 2\|\hat M - M\|/\sigma_K(M)$.  
\end{proof}
\textbf{Remark.}  The orthogonal matrix $Q$ will have no particular impact on the argument below. 
  For presentation simplicity, we assume, without loss of generality, that $Q=I$ in the rest
  of the proof.

   \begin{lemma}[Lemma 5.3 of \cite{LeiR14}]\label{lem:k-means}
     Let $\hat U$ and $U$ be two $n\times K$ matrices such that $U$
     contains $K$ distinct rows.  Let $\delta$ be the minimum
     distance between two distinct rows of $U$, and $g$ be the membership
     vector given by clustering the rows of $U$.  Let $\hat g$ be the output 
     of a $k$-means clustering algorithm on $\hat U$, with objective value
     no larger than a constant factor of the global optimum.
     Assume that $\|\hat U-U\|_{F}\le c n \delta$ for some small enough
     constant $c$.
     Then $\hat g$ agrees with $g$ on all but
     $c^{-1}\|\hat U - U\|_{F}\delta^{-1}$ nodes after an appropriate label permutation. 
   \end{lemma}

  For the proof of \Cref{thm:dcbm}, we need the analogous version of \Cref{lem:k-means} for the k-median algorithm, which is a simple adaptation of \Cref{lem:k-means}. For a matrix $M$, $\|M\|_{2,1}$ denotes the sum of
   $\ell_2$ norms of the rows in $M$.
   \begin{lemma}\label{lem:k-median}
     Let $\hat U$ and $U$ be two $n\times K$ matrices such that $U$
     contains $K$ distinct rows.  Let $\delta$ be the minimum
     distance between two distinct rows of $U$, and $g$ be the membership
     vector given by clustering the rows of $U$.  Let $\hat g$ be the output 
     of a
     $k$-median clustering algorithm on $\hat U$, with objective value
     no larger than a constant factor of the global optimum.
     Assume that $\|\hat U-U\|_{2,1}\le c n \delta$ for some small enough
     constant $c$ and that $g$ satisfies Assumption A2.
     Then $\hat g$ agrees with $g$ on all but
     $c^{-1}\|\hat U - U\|_{2,1}\delta^{-1}$ nodes after an appropriate label permutation. 
   \end{lemma}

\subsubsection{Proof of \Cref{thm:sbm-consist}}
\begin{proof}[Proof of \Cref{thm:sbm-consist}]
    Let $G$ be the $n\times K$ matrix with $G_{ij}=1$ if $j=g_i$, and 
    $G_{ij}=0$ otherwise.  Let $G^{(1)}$ be the submatrix of $G$ containing
    rows in $\mathcal N_1$.  Then
    $P^{(1)}=G^{(1)} B G^T = G \tilde B \tilde G^T$,
    where $\tilde G$  is an $n\times K$ matrix obtained by
    normalizing the columns of $G$. and $\tilde B$ is a $K\times K$
    matrix after corresponding column scaling of $B$.  It is easy to check that
    $\tilde G$ has orthonormal columns and hence the top $K$-dimensional right singular subspace of $P^{(1)}$
    is spanned by $U=\tilde G Q$, for any $K\times K$
    orthogonal matrix $Q$.  Thus $U$ contains $K$ distinct rows
    and the distance between two distinct rows is of order
    at least $1/\sqrt{n}$.
  
    We focus on the event that $\min_{k}n_{1,k}^*\ge \pi_0n/(2V)$ and
    $\|A^{(1)}-P^{(1)}\|\le C\sqrt{n\alpha_n}$, which has probability
    at least $1-n^{-1}$ according to \Cref{lem:split,lem:spec-bound}.
    Then it can be directly verified that $\sigma_K(P^{(1)})\ge C n\alpha_n$
    because the $K$th largest singular values of both $G^{(1)}$ and $G$
    are of order at least $\sqrt{n}$. Then
    \Cref{lem:svd-error} implies that
    $\hat U$ and $U$, the matrices consisting of $n\times K$ top singular
    vectors of $A^{(1)}$ and $P^{(1)}$, satisfy, with appropriate choice of 
    $Q$,
    $$
    \|\hat U - U\|_F^2\le C\frac{1}{n\alpha_n}\,.
    $$
    
  Then applying \Cref{lem:k-means}, we know that the $k$-means clustering
    algorithm misclusters no more than $C/\alpha_n$ nodes.
\end{proof}

\subsubsection{Proof of \Cref{thm:dcbm}}\label{sec:proof-dcbm}
\begin{proof}[Proof of \Cref{thm:dcbm}]
Let $\Psi$ be an $n\times K$ matrix such that $\Psi_{ij}=\psi'_i$ if
$j=g_i$ and $\Psi_{ij}=0$ otherwise. Let $\Psi^{(1)}$ be the corresponding
submatrix of $\Psi$ with rows in $\mathcal N_1$.  Then
$P^{(1)}=\Psi^{(1)} \tilde B \Psi$, where $\tilde B$ is a $K\times K$
matrix obtained after corresponding row/column scaling of $B$.
It is easy to check that $\Psi$ is orthonormal so that the top $K$-dimensional
right singular subspace of $P^{(1)}$ is spanned by $U =\Psi Q$ for any
$K\times K$ orthogonal $Q$.  It follows that the norm of
$i$th row of $U$ is $\psi'_i$, and that any two rows of $U$ 
in distinct communities are orthogonal. Let $\tilde U$ and $\tilde U^*$ be 
the row-normalized versions of $\hat U$ and $U$, respectively. 
Then $\tilde U^*$ contains $K$ distinct rows and the distance between
any two distinct rows of $\tilde U^*$ is $\sqrt{2}$.

We will focus on the event that $\min_{k}n_{1,k}^*\ge c_0\pi_0n/2$ and
  $\|A^{(1)}-P^{(1)}\|\le C\sqrt{n\alpha_n}$, which has probability
  at least $1-n^{-1}$ according to \Cref{lem:split,lem:spec-bound}.
Applying \Cref{lem:svd-error} we have,
$$
\|\hat U - U\|_F^2 \le C/(n\alpha_n)\,,
$$
 where we take the matrix $Q$ in \Cref{lem:svd-error} as identity.
Because the minimum row norm of $U$ is at least $\psi_0/\sqrt{n}$, the number of zero rows in $\hat U$ is at most
$\|\hat U - U\|_F^2 / (\psi_0/\sqrt{n})^2 = O(\alpha_n^{-1})=o(\sqrt{n/\alpha_n})$.  In the rest of the proof we can safely
assume that $\hat U$ has no zero rows.

Now let  $u_i$, $\hat u_i$ be the $i$th row of $U$, $\hat U$, respectively.
We have, using the fact that
$\left\|(u/\|u\| - v/\|v\|)\right\|\le 2\|u-v\|/\|v\|$ for all
vectors $u,v$ of same dimension, Cauchy-Schwartz, and
Assumption A3,
\begin{align*}
  \|\tilde U-\tilde U^*\|_{2,1}\le & 2\sum_{i=1}^n
  \frac{\|\hat u_i-u_i\|}{\|u_{i}\|}=
  2\sum_{i=1}^n\frac{\|\hat u_{i}-u_{i}\|}{\psi_i'}\\
  \le &2\|\hat U - U\|_F\left(\sum_{i=1}^n (\psi_i')^{-2}\right)^{1/2}
\le 2\psi_0^{-1}n\|\hat U - U\|_F\le C \sqrt{n/\alpha_n}\,.
\end{align*}
Applying \Cref{lem:k-median} to $\tilde U$ and $\tilde U^*$, we
conclude that
     $\hat g$ agrees with $g$ on all but
    $O(\sqrt{n/\alpha_n})$ nodes.

Recall that $\|u_i\|=\psi_i'$ 
for all $i$.
Then Cauchy-Schwartz implies that
\begin{align}\label{eq:psi-ell-1-err}
\|\hat\psi'-\psi'\|_1\le & \sum_{i=1}^n \|\hat u_i - u_i\| \le \sqrt{n}\|\hat U-U\|_F
\le C\alpha_n^{-1/2}\,.
\end{align}
By Assumptions (A2) and (A3) we have $\inf_{i}\psi_i'\ge Cn^{-1/2}$ for some constant $C$.

Let $S_n=\{i:|\hat\psi_i'-\psi_i'|\le n^{-1/2} (n^{1/3}\alpha_n)^{-1/2} \}$. Then
for all $i\in S_n$, we have $\hat\psi_i'=\psi_i'(1+o(1))$ and
\begin{equation}\label{eq:S_n}
|S_n^c|\le \frac{\|\hat\psi' - \psi'\|_1}{n^{-2/3}\alpha_n^{-1/2}}\le C n^{2/3}\,.
\end{equation}

For $1\le k<k'\le K$, consider the oracle estimator
$$
   \hat B_{k,k'}'^*=
    \frac{\sum_{i\in \mathcal N_{1,k}^*, j\in\mathcal N_{1,k'}^*\cup \mathcal N_{2,k'}^*}A_{ij}}
    {\sum_{i\in \mathcal N_{1,k}^*, j\in\mathcal N_{1,k'}^*\cup 
    \mathcal N_{2,k'}^*}\psi'_i\psi'_j}\,.
$$
It is obvious that $\hat P^*_{ij}=\psi_i'\psi_j' \hat B_{g_i 
g_j}'^*=(1+o_P(1))\psi_i\psi_jB_{g_i g_j}=(1+o_P(1))P_{ij}$.
As a result, the claim in part (b) of the theorem follows if we can show that
$\hat B_{k,k'}'=(1+o_P(1))\hat B_{k,k'}'^*$ because for all but a vanishing proportion
of pairs $(i,j)$ we have $(\hat g_i,\hat g_j)=(g_i, g_j)$ and $\hat\psi_i'
\hat\psi_j'=\psi_i'\psi_j'(1+o(1))$ in view of \eqref{eq:S_n}. To this end,
we compare
$$
n^{-1}\hat B_{k,k'}' = 
    \frac{\sum_{i\in \mathcal N_{1,k}, j\in\mathcal N_{1,k'}\cup \mathcal N_{2,k'}^*}A_{ij}}
    {\sum_{i\in \mathcal N_{1,k}, j\in\mathcal N_{1,k'}\cup 
    \mathcal N_{2,k'}}(\sqrt{n}\hat\psi'_i)(\sqrt{n}\hat\psi'_j)}
$$
with
$$
   n^{-1}\hat B_{k,k'}'^*=
    \frac{\sum_{i\in \mathcal N_{1,k}^*, j\in\mathcal N_{1,k'}^*\cup \mathcal N_{2,k'}^*}A_{ij}}
    {\sum_{i\in \mathcal N_{1,k}^*, j\in\mathcal N_{1,k'}^*\cup 
    \mathcal N_{2,k'}^*}(\sqrt{n}\psi'_i)(\sqrt{n}\psi'_j)}\,.
$$
Note that $\sqrt{n}\psi_i'\asymp 1$ for all $i$.  It is easy to check that the numerators
differ by at most $o(n^{5/3})$. For the denominators,
first we compare the denominator of $n^{-1}\hat B_{k,k'}'^*$
with
\begin{equation}\label{eq:denominator-interm}
\sum_{i\in \mathcal N_{1,k}, j\in\mathcal N_{1,k'}\cup 
    \mathcal N_{2,k'}}(\sqrt{n}\psi'_i)(\sqrt{n}\psi'_j)\,.
\end{equation}
It straightforward to check that their ratio tends to 1, because
the index sets of these two summations differ by a vanishing proportion.
Now compare \eqref{eq:denominator-interm} with the denominator of
$n^{-1}\hat B_{k,k'}'$. We have
\begin{align*}
\sum_{i\in \mathcal N_{1,k}, j\in\mathcal N_{1,k'}\cup 
    \mathcal N_{2,k'}}(\sqrt{n}\hat\psi'_i)(\sqrt{n}\hat\psi'_j)
    = &\left(\sum_{i\in\mathcal N_{1,k}} \sqrt{n}\hat\psi_i'\right)
    \left(\sum_{j\in\mathcal N_{1,k'}\cup \mathcal N_{2,k'}} \sqrt{n}\hat\psi_j'\right)\\
    =& (1+o(1))\left(\sum_{i\in\mathcal N_{1,k}} \sqrt{n}\psi_i'\right)
    \left(\sum_{j\in\mathcal N_{1,k'}\cup \mathcal N_{2,k'}} \sqrt{n}\psi_j'\right)\\
    =& (1+o(1))\sum_{i\in \mathcal N_{1,k}, j\in\mathcal N_{1,k'}\cup 
    \mathcal N_{2,k'}}(\sqrt{n}\psi'_i)(\sqrt{n}\psi'_j)\,,
\end{align*}
where the second line follows from the fact
$\sum_{i\in\mathcal N_{j,k}}\hat\psi_i'=(1+o(1))
\sum_{i\in\mathcal N_{j,k}}\psi_i'$ for all $j\in\{1,2\}$, $1\le k\le K$,
which is a consequence of \eqref{eq:psi-ell-1-err}.
Now we conclude that the denominators of $n^{-1}\hat B_{k,k'}'$
and $n^{-1}\hat B_{k,k'}'^*$
 are both of order
at least $n^2$ with ratio tending to one.  
Therefore, the absolute difference between $n^{-1}\hat B_{k,k'}'$
and $n^{-1}\hat B_{k,k'}'^*$ is $o(n^{-1/3})$
which is $o(n^{-1}\hat B_{k,k'}'^*)$.  The same argument can be used for the case $k=k'$.
\end{proof}

%


\bibliographystyle{apa-good}
\bibliography{network}
\end{document}